\documentclass{llncs}

\usepackage{listings}
\usepackage{makeidx}  
\usepackage{amsmath}
\usepackage{amsfonts}
\usepackage{bm}
\usepackage{stmaryrd}
\usepackage{caption}
\usepackage[svgnames]{xcolor}
\usepackage{todonotes}
\usepackage{thmtools}
\usepackage{thm-restate}
\usepackage[utf8]{inputenc}
\usepackage{xspace}

\usepackage{pgf}
\usepackage{tikz}
\usetikzlibrary{arrows}
\usetikzlibrary{automata}
\usetikzlibrary{petri}
\usetikzlibrary{shapes}
\usetikzlibrary{shapes.arrows}
\usetikzlibrary{calc}
\usetikzlibrary{fit}
\usetikzlibrary{backgrounds}
\usetikzlibrary{decorations.pathmorphing}
\usetikzlibrary{positioning}

\usepackage{hyperref}

\title{Monotonic Prefix Consistency in Distributed Systems}

\author{
    Alain Girault\inst{1} \and
    Gregor Gössler\inst{1} \and
    Rachid Guerraoui\inst{2} \and \\
    Jad Hamza\inst{3} \and
    Dragos-Adrian Seredinschi\inst{2}
} 

\institute{Univ.\ Grenoble Alpes, Inria, CNRS, Grenoble
  INP, LIG,
  38000 Grenoble, France \and LPD, EPFL \and LARA, EPFL}

\begin{document}

\pagestyle{headings}  


\newcommand{\set}[1]{\{{#1}\}}
\newcommand{\mset}[1]{\{\{{#1}\}\}}
\newcommand{\lem}[1]{Lemma~\ref{#1}}
\newcommand{\thm}[1]{Theorem~\ref{#1}}
\newcommand{\coro}[1]{Corollary~\ref{#1}}
\newcommand{\fig}[1]{Figure~\ref{#1}}
\newcommand{\sect}[1]{Section~\ref{#1}}
\newcommand{\apd}[1]{Appendix~\ref{#1}}
\newcommand{\tab}[1]{Table~\ref{#1}}
\newcommand{\ie}{i.e.~}
\newcommand{\resp}{resp.,~}
\newcommand{\foorp}{\hfill\ensuremath{\square}}

\newcounter{counttodos}
\newcommand{\ttodo}[1]{
  \stepcounter{counttodos}
  \todo[inline,color=blue!40]{
    TODO {\thecounttodos}: #1
  }
}

\newcommand{\Pid}{\mathbb{P}}
\newcommand{\Nat}{\mathbb{N}}
\newcommand{\PosNat}{\mathbb{N}^{>0}}
\newcommand{\updates}{{\sf Upd}}
\newcommand{\answers}{{\sf Ans}}
\newcommand{\Msg}{{\sf Msg}}
\newcommand{\Cid}{\mathbb{C}}
\newcommand{\option}[1]{{#1}^{\bot}}

\newcommand{\emptyseq}{\varepsilon}
\newcommand{\stronger}{\preceq}
\newcommand{\equivalent}{\approx}
\newcommand{\strictstronger}{\prec}

\newcommand{\history}{history}
\newcommand{\execution}{execution}
\newcommand{\trace}{trace}
\newcommand{\site}{site}
\newcommand{\sites}{sites}
\newcommand{\client}{client}
\newcommand{\apc}{{\sf APC}}
\newcommand{\hpc}{{\sf HPC}}
\newcommand{\hr}{history-recording}

\newcommand{\semantics}[1]{\llbracket {#1} \rrbracket}
\newcommand{\project}[2]{#1|_{#2}}
\newcommand{\fair}[1]{\textsf{fair}({#1})}

\newcommand{\hist}{h}
\newcommand{\ops}{O}
\newcommand{\opsof}[1]{O_{#1}}
\newcommand{\exec}{e}
\newcommand{\tr}{t}
\newcommand{\writeSet}{W}
\newcommand{\reads}{\tr_r}
\newcommand{\readsof}[1]{{#1}_r}
\newcommand{\pid}{\textit{pid}}
\newcommand{\pos}{\rho}
\newcommand{\rel}{\mathcal{R}}
\newcommand{\msg}{\textit{msg}}
\newcommand{\msgid}{\textit{mid}}
\newcommand{\meth}{r}
\newcommand{\rv}{a}
\newcommand{\view}{v}
\newcommand{\nsites}{n}
\newcommand{\aid}{\textit{aid}}
\newcommand{\cid}{c}
\newcommand{\act}{\sigma}
\newcommand{\initState}{q_{\textit{init}}}
\newcommand{\initStates}{\iota}

\newcommand{\getlabel}{\ell}
\newcommand{\gettrace}[1]{{\sf tr}({#1})}
\newcommand{\dc}[1]{{\downarrow {#1}}}
\newcommand{\send}[1]{{\tt broadcast}({#1})}
\newcommand{\receive}[1]{{\tt receive}({#1})}
\newcommand{\getround}{{\sf round}}
\newcommand{\growth}{{\sf growth}}
\newcommand{\consistent}{{\sf consistent}}
\newcommand{\replay}{{\sf replay}}
\newcommand{\takeprefix}[2]{{#1}^{\leq {#2}}}
\newcommand{\lastround}{{\sf lastround}}
\newcommand{\sra}{\shortrightarrow}

\newcommand{\chainvar}{{\tt chain}}
\newcommand{\bcvar}{{\tt bc}}
\newcommand{\noncevar}{{\tt nonce}}
\newcommand{\flabel}{\ell}

\newcommand{\axthinair}{{\sc ThinAir}}
\newcommand{\axrv}{{\sc RValue}}
\newcommand{\axprefix}{{\sc Prefix}}
\newcommand{\axbounded}{{\sc BDelay}}
\newcommand{\axmr}{{\sc MonotonicReads}}

\newcommand{\ltr}{<_\rel}
\newcommand{\leqrel}[1]{\leq_{#1}}
\newcommand{\ltrel}[1]{<_{#1}}
\newcommand{\po}{{\sf po}}
\newcommand{\hb}{{\sf hb}}
\newcommand{\ltpo}{<_{\po}}
\newcommand{\ltvis}{<_{\vis}}
\newcommand{\lthb}{<_{\hb}}
\newcommand{\leqp}{\ensuremath{\sqsubseteq}\xspace}
\newcommand{\ltp}{\ensuremath{\sqsubset}\xspace}

\newcommand{\vispast}[1]{{\sf past}({#1})}
\newcommand{\lasting}[1]{{\sf last}({#1})}
\newcommand{\thepast}[1]{{\sf \pi}_{#1}}
\newcommand{\pref}{\preceq_{pref}}
\newcommand{\weaker}{\preceq_{weak}}

\newcommand{\getbefore}[2]{{#1}^{#2}}

\newcommand{\pow}{PoW}
\newcommand{\impl}{\mathcal{I}}

\newcommand{\bcof}{{\sf BC_\view}}
\newcommand{\cli}{\mathcal{C}}
\newcommand{\spec}{S}
\newcommand{\speclist}{S_{\textit{list}}}
\newcommand{\op}{o}
\newcommand{\prefix}{P}
\newcommand{\loc}{loc}
\newcommand{\vis}{vis}
\newcommand{\loco}{\loc_\op}
\newcommand{\go}{go}
\newcommand{\dropped}{D}
\newcommand{\broadcast}{{\tt broadcast}}
\newcommand{\consistentprefix}{consistent monotonic prefixes}
\newcommand{\cmp}{{\sf CMP}}

\newcommand{\ret}{{\,\triangleright\,}}
\newcommand{\mkaction}[2]{{#1} \ret {#2}}

\newcommand{\projectrv}[2]{{#1}\{{#2}\}}
\newcommand{\getprefix}[2]{{#1}^{\leq {#2}}}

\newcommand{\crit}{\mathcal{C}}

\newcommand{\append}{{\tt write}}
\newcommand{\readlist}{{\tt read}}
\newcommand{\dat}{d}
\newcommand{\ok}{{\tt OK}}

\newcommand{\mkread}[1]{\mkaction{\readlist}{#1}}
\newcommand{\mkwrite}[1]{\append({#1})}

\newcommand{\handler}{{\sf msg\_handler}\xspace}
\newcommand{\updhandler}{{\sf update\_handler}\xspace}
\newcommand{\qryhandler}{{\sf query\_handler}\xspace}

\newcommand{\recreq}{R}
\newcommand{\self}{{\sf Event}}
\newcommand{\other}{{\sf History}}
\newcommand{\others}{{\sf Histories}}

\newcommand{\hasempty}{\textsf{HasEmpty}}
\newcommand{\availability}{\textsf{Availability}}
\newcommand{\commute}{\textsf{Commute}}
\newcommand{\invisibleread}{\textsf{InvisibleRead}}
\newcommand{\writesymmetry}{\textsf{WriterSymmetry}}

\tikzstyle{group}=[
  DarkRed,
  draw=DarkRed,
  fill=Red!10!White,
  thick,
  rounded corners=10 pt,
  inner sep=2pt,
  align=left
]

\tikzstyle{criterion}=[
  DarkBlue,
  draw=DarkBlue,
  fill=Honeydew,
  thick,
  rounded corners=10 pt,
  inner sep=7pt,
  align=left,
  transform shape=false
]

\tikzstyle{incomp}=[Red,decorate,decoration=snake,very thick]
\tikzstyle{critrel}=[<-,DarkGreen,very thick]

\newcommand{\MPC}{\texttt{MPC}\xspace}
\newcommand{\content}{\textsf{content}}
\newcommand{\appends}{\textsf{writes}}
\newcommand{\isAppend}{\textsf{isWrite}}
\newcommand{\isRead}{\textsf{isRead}}
\newcommand{\getPid}{\textsf{getPid}}

\newcommand{\completions}[1]{\textsf{Completions}({#1})}
\newcommand{\sequences}[1]{\textsf{Sequences}({#1})}

\newcommand{\var}{x}
\newcommand{\addOperation}[2]{\textsf{addOperation}({#1},{#2})}
\newcommand{\removeOperation}[2]{\textsf{removeOperation}({#1},{#2})}

\lstdefinelanguage{scala}{
  morekeywords={abstract,case,catch,class,def,%
    do,else,extends,false,final,finally,%
    for,if,implicit,import,match,mixin,%
    new,null,object,override,package,%
    private,protected,requires,return,sealed,%
    send,super,this,throw,trait,true,try,%
    type,val,var,while,with,yield},
  otherkeywords={=>,<-,<\%,<:,>:,\#,@},
  sensitive=true,
  morecomment=[l]{//},
  morecomment=[n]{/*}{*/},
  morestring=[b]",
  morestring=[b]',
  morestring=[b]"""
}

\lstset{
  keywordstyle=\color{blue}\ttfamily,
  commentstyle=\color{DarkGreen}\ttfamily,
  rulecolor=\color{black},
  mathescape=true,
  language=scala,
  numbersep=5pt,
  numberstyle=\scriptsize\color{DarkRed},
  numberblanklines=true,
  numbers=left,
  tabsize=2,
  frame={tb},
}

\newcommand{\labell}{\textsf{label}}

\newcommand{\needsrev}[1]{{\color{red}{#1}}}

\maketitle
 

\begin{abstract}
  We study the issue of data consistency in distributed
  systems. Specifically, we consider a distributed system that
  replicates its data at multiple sites, which is prone to partitions,
  and which is assumed to be available (in the sense that queries are
  always eventually answered). In such a setting, strong consistency,
  where all replicas of the system apply synchronously every
  operation, is not possible to implement. However, many weaker
  consistency criteria that allow a greater number of behaviors than
  strong consistency, are implementable in available distributed
  systems.

  \hspace*{3mm} We focus on determining the strongest consistency
  criterion that can be implemented in a convergent and available
  distributed system that tolerates partitions. We focus on objects
  where the set of operations can be split into updates and queries. We
  show that no criterion stronger than Monotonic Prefix Consistency
  (MPC) can be implemented. 

\end{abstract}




\section{Introduction}

\emph{Replication} is a mechanism that enables sites from different
geographical locations to access a shared data type with low
latency. It consists of creating copies of this data type on each
\emph{site} of a distributed system. Ideally, replication should be
transparent, in the sense that the users of the data type should
not notice discrepancies between the different copies of the data
type. 

An ideal replication scheme could be implemented by keeping all sites
synchronized after each update to the data type. This ideal model
is called \emph{strong consistency}, or
linearizability~\cite{journals/toplas/HerlihyW90}. The disadvantage of
this model is that it can cause large delays for users, and worse the
data type might not be \emph{available} to use at all times. This
may happen, for instance, if some sites of the system are unreachable,
i.e., partitioned from the rest of the network.  Briefly, it is not
possible to implement strong consistency in a distributed system while
ensuring \emph{high
availability}~\cite{bre12cap,DBLP:journals/sigact/GilbertL02,attiya2017limitations}.
High availability (hereafter \emph{availability} for short) means that
sites must answer users' requests directly, without waiting for outside
communication.

Given this impossibility, developers rely on weaker notions of
consistency, such as causal
consistency~\cite{Lamport:1978:TCO:359545.359563}. Weaker consistency
criteria do not require sites to be exactly synchronized as in strong
consistency. For instance, causal consistency allows different sites to
apply updates to the data type in different orders, as long as the
updates are not \emph{causally related}. Informally, a consistency
criterion specifies the \emph{behaviors} that are allowed by a replicated
data type. In this sense, causal consistency is more permissive than
strong consistency. We also say that strong consistency is
\emph{stronger} than causal consistency, as strong consistency allows
strictly fewer behaviors than causal consistency. A natural question is
then: What is the strongest consistency criterion that can be
implemented by a replicated data type? We focus in this paper on data
types where the set of operations can be split into two disjoint sets,
updates and queries. Updates modify the state and but do not return
values, while queries return values without modifying the state.

In~\cite{attiya2017limitations}, it was proven that nothing stronger
than \emph{observable causal consistency} (a variant of causal
consistency) can be implemented. It is an open question whether
observable causal consistency itself is actually implementable.
Moreover, \cite{attiya2017limitations} does not study consistency
criteria that are not comparable to observable causal consistency.
Indeed, there exist consistency criteria that are neither stronger
than causal consistency, nor weaker, and which can be implemented by a
replicated data type.

In our paper, we explore one such consistency criterion. More
precisely, we prove that, under some conditions which are natural in a
large distributed system (availability and convergence), nothing
stronger than \emph{monotonic prefix consistency}
(\MPC)~\cite{replicated-data-consistency-explained-through-baseball}
can be implemented. This result does not contradict the result
from~\cite{attiya2017limitations}, since \MPC and causal consistency
are incomparable.

The reason why \MPC and observable causal consistency are incomparable
is as follows. \MPC requires all sites to apply updates in the same
order (but not necessarily synchronized at the same time, as in strong
consistency), while causal consistency allows non-causally related
updates to be applied in different orders. On the other hand, causal
consistency requires all causally-related updates to be applied in an
order respecting causality, while \MPC requires no such constraint.


Overall, our contribution is to prove that, for a notion of behaviors
where the time and place of origin of updates do not matter,
nothing stronger than \MPC can be implemented in a
distributed setting. 
Moreover, we remark that clients that only have the observability
defined in Section~\ref{sec:problem} cannot tell the difference between
a strongly consistent implementation and an \MPC implementation.

In the rest of this paper, we first give preliminary notions and a
formal definition of the problem we are addressing (Sections
\ref{sec:implementations} and \ref{sec:problem}). We then turn our
attention to the \MPC model by defining it formally and through an
implementation (Section \ref{sec:mpc}). We prove that, given the 
observability mentioned above, and under conditions natural in a
large-scale network (availability, convergence), nothing stronger than
\MPC can be implemented (Section \ref{sec:mpcismin}). Then we compare
\MPC with other consistency models (Section \ref{sec:related}), and
conclude (Section \ref{sec:conclusion}).

To improve the presentation, some proofs are deferred to the appendix.










\section{Replicated Implementations}
\label{sec:implementations}

An \emph{implementation} of a replicated data type consists of
several \emph{sites} that communicate by sending messages. Messages are
delivered asynchronously by the network, and can be reordered or
delayed. To be able to build implementations that provide liveness
guarantees, we assume all messages are \emph{eventually} delivered by
the network.

Each site of an implementation maintains a local state. This local state
reflects the view that the site has on the replicated data type,
and may contain arbitrary data. Each site implements the protocol by
means of an \emph{update handler}, a \emph{query handler}, and a
\emph{message handler}.

The update handler is used by (hypothetical) clients to submit updates
to the data type. The update handler may modify the local states of
the site, and broadcast a message to the other sites. Later, when
another site receives the message, its \emph{message handler} is
triggered, possibly updating the local state of the site, and possibly
broadcasting a new message.

The \emph{query handler} is used by clients to make queries on the data
type. The query handler returns an answer to the client, and is a
read-only operation that does not modify the local state or broadcast
messages.

\begin{remark}
  Our model only supports broadcast and not general peer-to-peer
  communication, but this is without loss of generality. We can simulate
  sending a message to a particular site by writing the identifier of
  the receiving site in the broadcast message. All other sites would
  then simply ignore messages that are not addressed to them.
\end{remark}

In this paper, we consider implementations of the \emph{list data
  type}. The list supports an update operation of the form
$\append(\dat)$, with $\dat \in \Nat$, which adds the element $\dat$ to the end of the
list. The list also supports a query operation $\readlist$ that
returns the whole list $\ell \in \Nat^*$, which is a sequence of
elements in $\Nat$.

\begin{definition}
  \label{def:upd-ans}
  Let $\updates = \set{\append(\dat)\ |\ \dat \in \Nat}$ be the set of
  updates, and $\answers = \set{\readlist(\ell)\ |\ \ell \in \Nat^*}$ be
  the set of all possibles answers to queries.
\end{definition}


We focus on the list data type because queries return the history
of all updates that ever happened. In that regard, lists can encode any
other data type whose operations can be split in updates and
queries, by adding a processing layer after the query operation of the
list returns all updates. Data types that contain operations which
are queries and updates at the same time (e.g.~the Pop operation of a 
stack) are outside the scope of this paper. We now proceed to give the
formal syntax for implementations, and then the corresponding
operational semantics.


\begin{definition}
  \label{def:impl}
  An \emph{implementation} $\impl$ is a tuple \\
  $(Q,\initStates,\Pid,\Msg,\handler,\updhandler,\qryhandler)$ where
  \begin{itemize}
    \item $Q$ is a non-empty set of \emph{local states},

    \item $\Pid$ is a non-empty finite set of \emph{process identifiers},

    \item $\initStates: \Pid \rightarrow Q$ associates to each process
      an \emph{initial local state},

    \item $\Msg$ is a set of \emph{messages},

    \item $\handler: Q \times \Msg \rightarrow Q \times \option{\Msg}$
      is a total function, called the \emph{handler of incoming messages},
      which updates the local state of a site when a message is
      received, and possibly broadcasts a new message,

    \item $\updhandler: Q \times \updates \rightarrow Q \times
      \option{\Msg}$ is a total function, called the \emph{handler of
        updates}, which modifies the local state when an update is
      submitted, and possibly broadcasts a message.

    \item $\qryhandler: Q \rightarrow \answers$ is a total function, called
      the \emph{handler of client queries}, which returns an answer to
      client queries.
  \end{itemize}

  The set $\option{\Msg}$ is defined as $\Msg \uplus \set{\bot}$, where
  $\bot$ is a special symbol denoting the fact that no message is sent.
\end{definition}


Before defining the semantics of implementations, we introduce a few
notations. We first define the notion of an \emph{action}, used to denote
events that happen during the execution. Each action contains a unique
\emph{action identifier} $\aid \in \Nat$, and the process identifier
$\pid \in \Pid$ where the action occurs.

\begin{definition}
  \label{def:send-rec}
  A \emph{broadcast action} is a tuple
  $(\aid,\pid,\send{\msgid,\msg})$, and a \emph{receive action} is a
  tuple $(\aid,\pid,\receive{\msgid,\msg})$, where $\msgid \in \Nat$
  is the \emph{message identifier} and $\msg \in \Msg$ is the
  message. An \emph{update action} or a \emph{write action} is a tuple
  $(\aid,\pid,\append(d))$ where $d \in \Nat$. Finally, a \emph{query
    action} or a \emph{read action} is a tuple
  $(\aid,\pid,\readlist(\ell))$ where $\ell \in \Nat^*$.
\end{definition}

Executions are then defined as sequences of actions, and are considered 
up to action and message identifiers renaming.

\begin{definition}
  \label{def:exec}
  An \emph{execution} $\exec$ is a sequence of broadcast, receive,
  query and update actions where no two actions have the same
  identifier \aid, and no two broadcast actions have the same message
  identifier \msgid.
\end{definition}

We now describe how implementations operate on a given site $\pid
\in \Pid$.

\begin{definition}
  \label{def:follows}
  We say that a sequence of actions $\act_1 \dots \act_n \dots$ from
  site $\pid$ \emph{follows} $\impl$ if there exists a sequence of
  states $q_0 \dots q_n \dots$ such that $q_0 = \initStates(\pid)$, and
  for all $i \in \Nat \backslash \{0\}$, we have:
  \begin{enumerate}
    \item\label{update} if $\act_i = (\aid,\pid,\append(d))$ with $d \in
      \Nat$, then $\updhandler(q_{i-1}, \append(d)) = (q_i,\_)$.
      This means that upon a write action, a site must update its 
      state as defined by the update handler;


    \item\label{query} if $\act_i = (\aid,\pid,\readlist(\ell))$ with
      $\ell \in \Nat^*$, then $\qryhandler(q_{i-1}) = \readlist(\ell)$
      and $q_i = q_{i-1}$.
      This condition states that query actions do not modify the state, 
      and that the answer $\readlist(\ell)$ given to query actions 
      must be as specified by the query handler, depending on the 
      current state $q_{i-1}$;

    \item\label{broadcast} if $\act_i =
      (\aid,\pid,\send{\msgid,\msg})$, then $q_i = q_{i-1}$.
      Broadcast actions do not modify the local state;

    \item\label{receive} if $\act_i =
      (\aid,\pid,\receive{\msgid,\msg})$, then $\handler(q_{i-1},\msg)
      = (q_i,\_)$.
      The reception of a message modifies the local state as specified 
      by $\handler$.


    \newcounter{enumTemp}
    \setcounter{enumTemp}{\theenumi}
  \end{enumerate}

  Moreover, we require that broadcast actions are performed if and only
  if they are triggered by the handler of incoming messages, or the
  handler of clients requests. Formally, for all $i > 0$, $\act_i =
  (\aid,\pid,\send{\msgid,\msg})$ if and only if either:
  \begin{enumerate}
    \setcounter{enumi}{\theenumTemp}

    \item $\exists\, \append(\dat) \in \updates$ and $\aid\,' \in \Nat$
      such that $\act_{i-1} = (\aid\,',\pid,\append(\dat))$ and
      \\ $\updhandler(q_{i-1},\append(\dat)) = (q_i,\msg)$, or

    \item $\exists\, \aid\,' \in \Nat$, $\msgid \in \Nat$, and $\msg\,'
      \in \Msg$ such that \\ $\act_{i-1} =
      (\aid\,',\pid,\receive{\msgid,\msg})$ and
      $\handler(q_{i-1},\msg\,') = (q_i,\msg)$.
  \end{enumerate}

  When all conditions hold, we say that $q_0 \dots q_n \dots$ is a
  \emph{run} for $\act_1 \dots \act_n \dots$. Note that when a run 
  exists for a sequence of actions, it is unique.
\end{definition}




We then define the set of executions generated by~$\impl$,
denoted~$\semantics{\impl}$. In particular, this definition models the
communication between sites, and specifies that a receive action may
happen only if there exists a broadcast action with the same message
identifier preceding the receive action in the execution. Moreover, a
\emph{fairness} condition ensures that, in an infinite execution, every
broadcast action must have a corresponding receive action on every site.

\begin{definition}
  \label{def:set-of-exec}
  Let $\impl$ be an implementation. The set of executions generated by
  $\impl$ is $\semantics{\impl}$ such that $e \in \semantics{\impl}$
  if and only if the three following conditions hold:
  \begin{itemize}
    \item \textbf{Projection:} for all $\pid \in \Pid$, the projection
      $\project{\exec}{\pid}$ follows~$\impl$,

    \item \textbf{Causality:} for every receive action $\act =
      (\aid,\pid,\receive{\msgid,\msg})$, there exists a broadcast
      action $(\aid\,',\pid\,',\send{\msgid,\msg})$ before $\act$
      in~$\exec$,

    \item \textbf{Fairness:} if $\exec$ is infinite, then for every
      site $\pid \in Pid$ and every broadcast action
      $(\aid\,',\pid\,',\send{\msgid,\msg})$ performed on any
      site~$\pid\,'$, there exists a receive action
      $(\aid,\pid,\receive{\msgid,\msg})$ in~$\exec$,
  \end{itemize}
  where $\project{\exec}{\pid}$ is the subsequence of $\exec$ of
  actions performed by process~$\pid$:
  \begin{itemize}
    \item $\project{\varepsilon}{\pid} = \varepsilon$;

    \item $\project{((\aid,\pid,x).\exec)}{\pid} =
    (\aid,\pid,x).(\project{\exec}{\pid})$;

    \item $\project{((\aid,\pid\,',x).\exec)}{\pid} =
    \project{\exec}{\pid}$ whenever $\pid\,' \neq \pid$.
  \end{itemize}
\end{definition}

\begin{remark}
  The implementations we consider are \emph{available} by construction,
  in the sense that any site allows any updates or queries to be done at
  any time, and answers to queries directly. This is ensured by the fact
  that our update and query handlers are total functions. More
  precisely, the item \ref{update} of Definition~\ref{def:follows}
  (together with Definition~\ref{def:set-of-exec}) ensures that updates
  can be performed at any time through the update handler (\emph{update
  availability}). 
  
  The broadcast action that happens right after an update action must be
  thought of as happening right after the update. Broadcast actions do not
  involve actively waiting for responses, and as such do not prevent
  availability.
  
  Similarly, the item \ref{query} of
  Definition~\ref{def:follows} ensures that any query of any site is
  answered immediately, only using the local state of the site
  (\emph{query availability}). We later formalize this in
  Lemmas~\ref{lemma:updavailability} and \ref{lemma:availability}.
\end{remark}

For the rest of the paper, we consider that updates are unique, in the
sense that an execution may not contain two update actions that write
the same value $d \in \Nat$. This assumption only serves to simplify the
presentation of our result, and can be done without loss of generality,
as updates can be made unique by attaching a unique timestamp to them.


\section{Problem Definition}
\label{sec:problem}

In this section, we explain how we compare implementations using the
notion of a \emph{trace}. Informally, the trace of an execution
corresponds to what is observable from the point of view of clients
using the data type.

Our notion of a trace is based on two assumptions:
(1)~Clients know the order of the queries they have done on a
    site, but not the relative positions of their queries with respect
    to other clients' queries.
(2)~The origin of updates is not relevant from a client's perspective. This
    models publicly accessible data structures where any client can disseminate a
    transaction in the network, and the place and time where the
    transaction was created are not relevant for the protocol execution.

More precisely, a trace records an unordered set of
updates (without their site identifiers), and records for each site the
sequence of queries that happened on this site.

\begin{definition}
\label{sec:deftrace}
  A \emph{trace} $(\reads,\writeSet)$ is a pair where $\reads$ is a
  labelled partially ordered set (see hereafter for more details), 
  and $\writeSet$ is a subset of $\Nat$.
  The trace $(\reads,\writeSet)$ corresponding to an execution
  $\exec$ is denoted $\gettrace{\exec}$, where
  $\reads = (A,<,\labell)$ is a labelled partially ordered set such
  that:
  \begin{itemize}
    \item $A$ is the set of action identifiers of query actions of
      $\exec$;

    \item $<$ is a transitive and irreflexive relation over $A$,
      sometimes called the \emph{program order}, ordering queries
      performed on the same site; more precisely, we have $\aid <
      \aid\,'$ if $\aid,\aid\,' \in A$ are action identifiers
      performed by the same site, and that appear in that order in~$e$;

    \item $\labell: A \rightarrow \answers$ is the \emph{labelling
      function} such that for any $\aid \in A$, $\labell(\aid)$ is the
      answer of the query action corresponding to $\aid$ in~$\exec$;
  \end{itemize}
  and $\writeSet \subseteq \Nat$ is the set of elements that appear in
  an update action of~$\exec$.
\end{definition}

\begin{example}
  \label{example} Consider the execution $e$ in~\fig{fig:trace}, and its
  corresponding trace $\gettrace{e}$. ($\pid_1,\pid_2,\pid_3 \in \Pid$
  are site identifiers, $\msgid_1,\allowbreak \msgid_2,\msgid_3 \in
  \Nat$ are unique message identifiers, and $\msg_1,\msg_2,\msg_3 \in
  \Msg$ are messages.) 
\end{example}

\tikzstyle{ordering}=[->,thick]
\tikzstyle{pointy}=[circle,fill=Black,inner sep=0, minimum width=4pt]

\newcommand{\place}[5]{
    \node[pointy] at (#1,#2) ({#4}) {};
    \node[yshift=3mm] at (#4) { #3 };
    \node[yshift=-3mm] at (#4) { #5 };
}

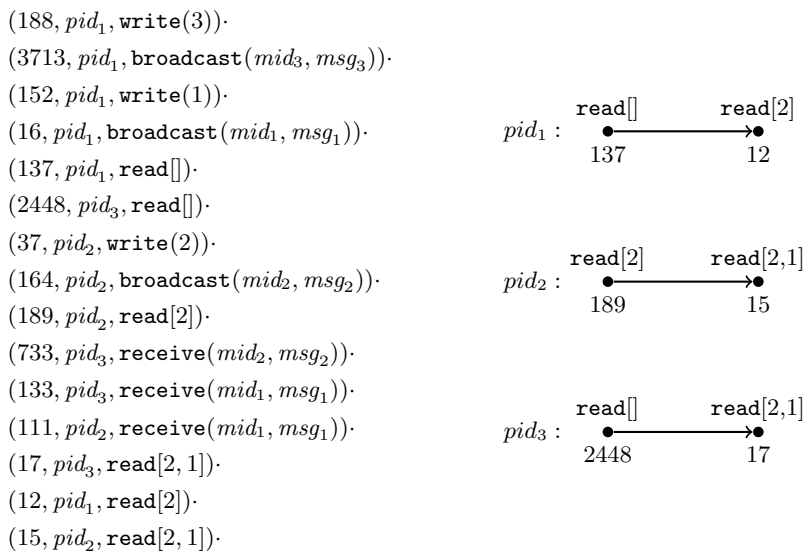
\begin{figure}[htb]
  \begin{minipage}[t]{0.59\textwidth}
  \begin{flalign*}
  &  (188, \pid_1, \append(3)) \cdot \\
  &  (3713, \pid_1, \send{\msgid_3,\msg_3}) \cdot \\
  &  (152, \pid_1, \append(1)) \cdot \\
  &  (16, \pid_1, \send{\msgid_1,\msg_1}) \cdot \\
  &  (137, \pid_1, \readlist []) \cdot \\
  &  (2448, \pid_3, \readlist []) \cdot \\
  &  (37, \pid_2, \append(2)) \cdot \\
  &  (164, \pid_2, \send{\msgid_2,\msg_2}) \cdot \\
  &  (189, \pid_2, \readlist [2]) \cdot \\
  &  (733, \pid_3, \receive{\msgid_2,\msg_2}) \cdot \\
  &  (133, \pid_3, \receive{\msgid_1,\msg_1}) \cdot \\
  &  (111, \pid_2, \receive{\msgid_1,\msg_1}) \cdot \\
  &  (17, \pid_3, \readlist [2,1]) \cdot \\
  &  (12, \pid_1, \readlist [2]) \cdot \\
  &  (15, \pid_2, \readlist [2,1]) \cdot
  \end{flalign*}
  \end{minipage} %
  \begin{minipage}[t]{0.4\textwidth}
  \raisebox{-1.5cm}{
  \raisebox{-\height}{
    \begin{tikzpicture}[xscale=4,yscale=2]
      \node[] at (1,0) () {$pid_1:$};
      \place{1.25}{0}{\readlist []}{C}{137};
      \place{1.75}{0}{\readlist [2]}{D}{12};
      \node[] at (1,-1) () {$pid_2:$};
      \place{1.25}{-1}{\readlist [2]}{E}{189};
      \place{1.75}{-1}{\readlist [2,1]}{F}{15};
      \node[] at (1,-2) () {$pid_3:$};
      \place{1.25}{-2}{\readlist []}{G}{2448};
      \place{1.75}{-2}{\readlist [2,1]}{H}{17};
      \draw[ordering] (C) -- (D);
      \draw[ordering] (E) -- (F);
      \draw[ordering] (G) -- (H);
    \end{tikzpicture}
  }
  }
  \end{minipage}
  
  \caption{An execution $e$ read from top to bottom ($188, \dots, 15$)
    and its corresponding trace $\gettrace{e} = (\reads,\writeSet)$
    (right). The bullets represent the action identifiers of~$\reads$
    (written under the bullet), and the corresponding labels are
    represented right above. The arrows represent the program order
    $<$ of $\reads$. The set of writes is $\writeSet=\set{1,2,3}$
    (from actions $152$, $37$, and $188$
    respectively).}  \label{fig:trace}
\end{figure}

Then, we compare implementations by looking at the set of traces they
produce. The fewer traces an implementation produces, the stronger it
is, and the closer it is to strong consistency.

\begin{definition}
  \label{def:stronger}
  The notation $\gettrace{}$ is extended to sets of executions
  point-wise. An implementation $\impl_1$ is \emph{stronger} than
  $\impl_2$, denoted $\impl_1 \stronger \impl_2$ iff
  \begin{equation*}
    \gettrace{\semantics{\impl_1}} \subseteq
    \gettrace{\semantics{\impl_2}}
  \end{equation*}

  The implementations $\impl_1$ and $\impl_2$ are said to be
  \emph{equivalent}, denoted $\impl_1 \equivalent \impl_2$, iff
  $\impl_1 \stronger \impl_2$ and $\impl_2 \stronger
  \impl_1$. Moreover, $\impl_1$ is \emph{strictly stronger} than
  $\impl_2$, denoted $\impl_1 \strictstronger \impl_2$, iff $\impl_1
  \stronger \impl_2$ and $\impl_1 \not\equivalent \impl_2$.
\end{definition}

Our goal is to find an implementation $\impl$ which is minimal in the
$\stronger$ ordering, i.e.,~for which there does not exist an
implementation $\impl'$ strictly stronger than~$\impl$.


\section{Definition of Monotonic Prefix Consistency (\MPC)}
\label{sec:mpc}


Often called consistent
prefix~\cite{replicated-data-consistency-explained-through-baseball,guerraoui2016trade},
the \MPC model requires that all sites of the replicated system agree
on the order of write operations (i.e., updates on the state). More
precisely, this means that given two read operations (possibly on two
different sites), one read has to return a list of writes which is a
prefix of the other. Moreover, read operations which execute on the
same site are monotonic. This means that subsequent reads at the same
site reflect a non-decreasing prefix of writes, i.e., the prefix must
either increase or remain unchanged. The trace given in
\fig{fig:trace} satisfies these constraints.

Note that the order on write operations on which the sites agree does
not necessarily satisfy causality among these operations nor
real-time. In other words, the order in which clients submit write
operations does not translate into any constraints on the order in
which these updates apply at all sites.  Moreover, \MPC does not
guarantee that a read operation will return \emph{all} of the
preceding writes, only a prefix of these writes. For instance, some
sites can be later than other sites in applying some updates.


\begin{definition}
  \label{def:prefix}
  Given two lists $\ell_1,\ell_2 \in \Nat^*$, we say that $\ell_1$ is
  a \emph{prefix} of $\ell_2$, denoted $\ell_1 \leqp \ell_2$, if there
  exists $\ell_3 \in \Nat^*$ such that $\ell_2 = \ell_1 \cdot \ell_3$.
  Moreover, $\ell_1$ is a \emph{strict prefix} of $\ell_2$, denoted 
  $\ell_1 \ltp \ell_2$, if $\ell_1 \leqp \ell_2$ and $\ell_1 \neq \ell_2$.
\end{definition}

By abuse of notation, we extend the prefix order to elements of
\answers, which are of the form $\readlist(\ell)$ where $\ell$ is a list
(see Def.~\ref{def:upd-ans}). Moreover, we also use the prefix notations
for other types of sequences, such as executions. We now formally define
\MPC.

\begin{definition}
  \label{def:mpc}
  \MPC is the set of traces $(\reads,\writeSet)$ where $\reads =
  (A,<,\labell)$ satisfying the following conditions:
  \begin{itemize}
  \item \textbf{Monotonicity:} A query $\aid\,'$ done after $\aid$ on
    the same site cannot return a smaller list. For all $\aid,\aid\,'
    \in A$, if $\aid < \aid\,'$, then $\labell(\aid) \leqp
    \labell(\aid\,')$.

  \item \textbf{Prefix:} Queries done on different sites are
    compatible, in the sense that one is a prefix of the other. For
    any all $\aid,\aid\,' \in A$, $\labell(\aid) \leqp
    \labell(\aid\,')$ or $\labell(\aid\,') \leqp \labell(\aid)$.

  \item \textbf{Consistency:} Queries only return elements that
    come from a write. For all $\aid \in A$, and for any element $d
    \in \Nat$ of $\labell(\aid)$, we have $d \in \writeSet$.
  \end{itemize}
\end{definition}


\section{Feasibility of MPC}
\label{sec:mpc-feasible}




\begin{figure}[t]

\scriptsize

\begin{lstlisting}
// Each site stores an element of Q, defined as a list of numbers
type Q = List[Nat]

abstract class Msg
// Forwarded messages go from Site i to Site 1, for all i > 1
case class Forwarded(d: Nat) extends Msg
// Apply messages originate from Site 1 and go to Site i, for i > 1
case class Apply(d: Nat) extends Msg

// The update handler for Site 1 appends element `upd' to q,
// and tells the other sites to do the same with Apply(upd)
def update_handler(q: Q, upd: Upd) = (append(q,upd), Apply(upd))

// The update handler for Site i > 1 sends a message Forwarded(upd)
// which is destined for Site 1, and does not modify the state
def update_handler(q: Q, upd: Upd) = (q, Forwarded(upd))

// Message handler for Site 1 (ignores Apply messages)
def msg_handler(msg: Msg) = msg match {
  case Forwarded(d) => (append(q,d), Apply(upd))
}

// Message handler for Site i > 1 (ignores Forwarded messages)
def msg_handler(msg: Msg) = msg match {
  case Apply(d) => (append(q,d), $\bot$)
}

// The query handler of any site returns the local state
def query_handler(q: Q) = q
\end{lstlisting}

\caption{An implementation of \MPC which is centralized at Site $1$.}
\label{fig:mpc}
\end{figure}

In this section, we provide a toy implementation (\fig{fig:mpc}) whose
traces are all in \MPC, to show that \MPC is indeed implementable. The
idea is to let Site~$1$ decide on the order of all update
operations. In general, the consensus mechanism for implementing \MPC
can be arbitrary, and symmetric with respect to sites, but we present
this one for its simplicity.

For ease of presentation, we assume here that update and message
handlers can be different depending on the site. This can be simulated
in our original definition by using the $\initStates$ function
(Def.~\ref{def:impl}, Section~\ref{sec:implementations}), which
defines a particular initial state for each site

Each site maintains a local state (in $Q$) which is the prefix of
updates as decided by Site~$1$. Upon receiving an update (line~16),
Site~$i$ with $i > 1$ forwards the update to Site~$1$. When receiving
an update (line~12) or when receiving a forwarded message (line~20),
Site~$1$ updates its local state, and broadcasts an \texttt{Apply}
messages for the other sites. Finally, when receiving
an \texttt{Apply} messages (line~25), Site $i$ with $i > 1$, updates
its local state.

We assume that the \texttt{Apply} messages sent by Site $1$ are
received in the same order in which they are sent, which can be
implemented by having Site $1$ add a local version number to each
broadcast message, and having sites with $i > 1$ cache messages until
all previous messages have been received. Similarly, we assume that
each message which is sent by a site is treat at most once by each of
the other sites. We omit these details in \fig{fig:mpc}. Finally, the
query handler of each site (line~29) simply returns the list
maintained in the local state.

We now prove that all the traces of the implementation described in
Figure~\ref{fig:mpc} satisfy \MPC. 

\begin{restatable}{proposition}{feasability}
  Let $\impl$ be the implementation of \fig{fig:mpc}. Then $\impl
  \stronger\MPC$.
\end{restatable}

The formal proof is in Appendix~\ref{ann:MPC}. It relies on the
observation that the implementation maintains the following invariant:
\begin{itemize}
\item (Related to \textbf{Monotonicity}) The list maintained in the
  local state $Q$ of each site grows over time.

\item (Related to \textbf{Prefix}) At any moment, given two lists
  $\ell_1$ and $\ell_2$ of two sites, $\ell_1$ is a prefix of $\ell_2$
  or vice versa. Any list is always a prefix of (or equal to) the list
  of Site~$1$.

\item (Related to \textbf{Consistency}) The list of a site only
  contains values that come from some update.
\end{itemize}


\section{Nothing Stronger Than \MPC in a Distributed Setting}
\label{sec:mpcismin}

We now proceed to our main result, stating that there exists no
\emph{convergent} implementation stronger than \MPC. Convergent in our
setting means that every write action performed should \emph{eventually}
be taken into account by all sites. We formalize this notion in
Section~\ref{ssec:convergence}. This convergence assumption prevents
trivial implementations, for instances ones that do not communicate and
always return the empty list for all queries.

In Section~\ref{ssec:props}, we prove several lemmas that hold for all
implementations. We make use of these lemmas to prove our main theorem
in Section~\ref{ssec:proof}.


\subsection{Convergence Property}
\label{ssec:convergence}

Convergence is formalized using the notion of eventual consistency (see
e.g.~\cite{principles-of-eventual-consistency,DBLP:conf/popl/BouajjaniEH14}
for definitions similar to the one we use there). A trace is eventually
consistent if every write is \emph{eventually} propagated to all sites.
More precisely, for every action $\append(d)$, the number of queries
that do not contain $d$ in their list must be finite. Note that this
implies that all finite traces are eventually consistent. 

\begin{definition}
  \label{def:eventual}
  A trace $(\reads,\writeSet)$ with $\reads = (A,<,\labell)$ is 
  \emph{eventually consistent} if for every $d \in \writeSet$, 
  the set \( \set{\aid \in A\ |\ d \not \in \labell(\aid)} \)
  is finite.
  An implementation is \emph{convergent} if all of its traces are
  eventually consistent.
\end{definition}


\subsection{Properties of Implementations}
\label{ssec:props}

Lemmas~\ref{lemma:updavailability}, \ref{lemma:availability},
and~\ref{lemma:invisiblereads} describe basic closure properties of
the set of executions generated by implementations in our setting. The
semantics described in Section~\ref{sec:implementations} ensures that
new updates and queries can always be performed following an existing
execution.  Moreover, queries never modify the state, and therefore
removing a read action from an execution does not affect its validity
(Lemma~\ref{lemma:invisiblereads}).

\begin{lemma}[Update Availability]
  \label{lemma:updavailability}
  Let $\impl$ be an implementation. Let $\exec$ be a finite execution
  in $\semantics{\impl}$, and let $(\reads,\writeSet) =
  \gettrace{\exec}$. Let $d \in \Nat$. Then, there exists an execution
  $e' \in \semantics{\impl}$ such that $e$ is a prefix of $e'$ and
  $\gettrace{e'} = (\reads,\writeSet \cup \set{d})$.
\end{lemma}

\begin{proof}
  Since $e \in \semantics{\impl}$, we know by
  Definitions~\ref{def:follows} and~\ref{def:set-of-exec} that
  $\project{e}{\pid}$ follows~$\impl$ and that there exists a run
  $q_0,\dots,q_n$ for $\project{e}{\pid}$. Let $(q_{n+1}, \msg) =
  \updhandler(q_n, \append(d))$. We distinguish two cases:
%

  (1) If $\msg = \bot$, let $e' = e \cdot (\aid,\pid,\append(d))$,
  where $\aid \in \Nat$ is a fresh action identifier that does not
  appear in~$e$, and $\pid$ is any process identifier in $\Pid$.


  (2) If $\msg \in \Msg$, let $e' = e \cdot (\aid_1,\pid,\append(d))
  \cdot (\aid_2,\send{\msgid,\msg})$, where $\aid_1$, $\aid_2$ are fresh
  action identifiers, and $\msgid$ is a fresh message identifier.

  In both cases, we construct a new run by adding the state $q_{n+1}$ at
  the end of the run $q_0,\dots,q_n$ (once in case~1, and twice in
  case~2). By Definition~\ref{def:follows}, this ensures that
  $\project{e'}{\pid}$ follows~$\impl$, and we then obtain $e' \in
  \semantics{\impl}$ by Definition~\ref{def:set-of-exec}. Moreover, we
  have $\gettrace{e'} = (\reads,\writeSet \cup \set{d})$, which
  concludes our proof. \foorp
\end{proof}

The next lemma shows that the implementation is \emph{available for
  queries}. This means that given a finite execution, we can perform a
query on any site and obtain an answer, as ensured by the definitions
given in Section~\ref{sec:implementations}. The proof is in
Appendix~\ref{ann:closure}.

\begin{restatable}[Query Availability]{lemma}{queryavailability}
  \label{lemma:availability} Let $\impl$ be an implementation. Let $e
  \in \semantics{\impl}$ be a finite execution and $\pid \in \Pid$.
  Then, there exist $\aid \in \Nat$ and $\ell \in \Nat^*$ such that the
  execution $e' = e \cdot (\aid,\pid,\readlist(\ell))$ belongs to
  $\semantics{\impl}$.
\end{restatable}

We then prove it is possible to remove any query action from an
execution.

\begin{restatable}[Invisible Reads]{lemma}{invisiblereads}
  \label{lemma:invisiblereads} Let $\impl$ be an implementation. Let
  $\exec \in \semantics{\impl}$ be an execution (finite or infinite) of
  the form $\exec_1 \cdot (\aid,\pid,\readlist(\ell)) \cdot
  \exec_2$, where $\aid \in \Nat$, $\pid \in \Pid$ and $\ell \in
  \Nat^*$. Then, $\exec_1 \cdot \exec_2 \in \semantics{\impl}$.
\end{restatable}

Lemma~\ref{lemma:limit} shows that, given an infinite sequence of
increasing finite executions $e_1\dots,e_n,\dots$ that satisfy a
fairness condition, the \emph{limit} execution (which is infinite) also
belongs to $\semantics{\impl}$. The fairness condition states that each
broadcast that appears in an execution $e_i$ must have corresponding
receive actions for each of the other sites $\pid \in \Pid$ in some
executions $e_j$.

\begin{definition}
  \label{definition:limit}
  Given an infinite sequence of finite sequences $e_1\dots,e_n,\dots$,
  such that for all $i \geq 1$, $e_i \ltp e_{i+1}$, the \emph{limit}
  $e^\infty$ of $e_1\dots,e_n,\dots$ is the (unique) infinite sequence
  such that for all $i$, $e_i \ltp e^\infty$.
\end{definition}

\begin{restatable}[Limit]{lemma}{limitexec}
  \label{lemma:limit} Let $\impl$ be an implementation. Let
  $e_1\dots,e_n,\dots$ be an infinite sequence of finite executions,
  such that for all $i \geq 1$, $e_i \in \semantics{\impl}$, $e_i
  \ltp e_{i+1}$, and such that for all $i \geq 1$, for all broadcast
  actions in $e_i$, and for all $\pid \in \Pid$, there exists $j \geq 1$
  such that $e_j$ contains a corresponding receive action.

  Then, the limit $e^\infty$ of $e_1\dots,e_n,\dots$ belongs to 
  $\semantics{\impl}$.
\end{restatable}

We finally prove in Lemma~\ref{lemma:convergence} that, given any
finite execution~$e$, it is possible to add a query action that
returns a list containing all the elements $W$ appearing in some write
action of~$e$. The proof relies on extending $e$ into an infinite
execution $e^\infty$ with an infinite number of queries. Our
convergence assumption then ensures that only finitely many of those
queries can ignore~$W$ (that is, return a list that does not contain
all elements of~$W$). This shows that there exists a query operation
(actually, infinite many) in $e^\infty$ that returns a list containing
all elements of~$W$. We can therefore take the finite prefix of
$e^\infty$ that ends with this query operation.

\begin{lemma}[Convergence]
  \label{lemma:convergence} Let $\impl$ be a convergent implementation.
  Let $e \in \semantics{\impl}$ be a finite execution and $\pid \in
  \Pid$. Let $W \subseteq \Nat$ be the set of elements appearing in an
  update action of~$e$, i.e.,
  $W = \set{d \in \Nat\ |\ \exists (\aid,\pid,\append(d)) \in \exec}$.

  Then, $e$~can be extended in an execution $e \cdot e' \cdot
  (\aid,\pid,\readlist(\ell)) \in \semantics{\impl}$ where $\ell \in
  \Nat^*$ contains every element of $W$, i.e., $W \subseteq \set{d \in
    \Nat \ |\ d \in \ell}$.  Moreover, we can define such an extension
  $e'$ that does not contain any query or update actions.
\end{lemma}

\begin{proof}
  We build an infinite sequence of finite executions
  $e_1,\dots,e_n,\dots$, where for every $i \geq 1$, $e_i \in
  \semantics{\impl}$. Moreover, we have $e_1 = e$ and for every $i
  \geq 1$, $e_i \leqp e_{i+1}$, and $e_{i+1}$ is obtained from $e_i$
  as follows.
  
  For every broadcast action $(\aid_1,\pid_1,\send{\msgid,\msg})$
  in~$e_i$, and for every $\pid_2 \in \Pid$, if there is no receive
  action $(\_,\pid_2,\receive{\msgid,\msg})$ in~$e_i$, then we add one
  when constructing~$e_{i+1}$. Moreover, if the message handler
  specifies that a message $\msg\,'$ should be sent when $\msg$ is
  received, we add a new broadcast action that sends $\msg\,'$,
  immediately following the receive action.  Finally, using
  Lemma~\ref{lemma:availability}, we add a query action (\readlist) on
  site~$\pid$.

  Then, we define $e^\infty$ to be the limit of $e_1,\dots,e_n,\dots$
  By Lemma~\ref{lemma:limit}, we have $e^\infty \in
  \semantics{\impl}$. Since $\impl$ is convergent, we know that
  $e^\infty$ is eventually consistent. This ensures that for every $d
  \in W$, out of the infinite number of queries that belong
  to~$e^\infty$, only finitely many do not contain~$d$.

  Therefore, there exists $i \geq 1$ such that $e_i$ ends with a query
  action that contains every element of~$W$. By construction, $e_i$~is
  of the form $e \cdot e'' \cdot (\aid,\pid,\readlist(\ell))$. Using
  Lemma~\ref{lemma:invisiblereads}, we remove every query action that
  appears in~$e''$, and obtain an execution of the form $e \cdot e'
  \cdot (\aid,\pid,\readlist(\ell))$ where $\ell \in \Nat^*$ contains
  every element of~$W$, and where $e'$ does not contain any query or
  update actions. \foorp
\end{proof}


\subsection{Nothing Is Stronger Than \MPC in a Distributed Setting}
\label{ssec:proof}

We now proceed with the proof that no convergent implementation is
strictly stronger than \MPC. We start with an implementation $\impl$
that is strictly stronger than \MPC and derive a contradiction.

More precisely, using the lemmas proved in Section~\ref{ssec:props}, we prove
that any trace of \MPC belongs to $\gettrace{\semantics{\impl}}$. First, we show
in Lemma~\ref{lem:finitetrace} that this holds for finite traces, by using an
induction on the number of write operations in the trace. For each write
operation $w$, we apply Lemma~\ref{lemma:convergence} in order to force the
sites to take into account $w$.

\begin{lemma}
  \label{lem:finitetrace}
  Let $\impl$ be a convergent implementation such that $\impl
  \strictstronger \MPC$, and let $\tr$ be a finite trace of \MPC. Then,
  there is a finite execution $\exec \in \semantics{\impl}$ such that
  $\gettrace{\exec} = \tr$.
\end{lemma}

\begin{proof}
  Let $\tr = (\reads,\writeSet)$. We proceed by induction on the size
  of~$\writeSet$, denoted~$n$.
 
  \textbf{Case $\bm{n = 0}$.} In that case, the set $\writeSet$ is empty.
  First, by definition of $\semantics{\impl}$, we have $\emptyseq \in
  \semantics{\impl}$ where $\emptyseq$ is the empty execution. Then, for
  each read operation in $\tr$, and using
  Lemma~\ref{lemma:availability}, we add a read operation to the 
  execution. We obtain an execution $e \in \semantics{\impl}$.
  
  We then have to prove that $\gettrace{e} = \tr$, meaning that all the
  read operations of $e$ return the empty list, as in $\tr$. 
  By our assumption that $\impl \strictstronger \MPC$, we know that
  $\gettrace{e} \in \MPC$. 
  By definition of $\MPC$, and since $e$ contains no write operation,
  the Consistency property of $\MPC$ ensures that all the read actions
  of $e$ return the empty list. Therefore, we have $\gettrace{e} = \tr$,
  which concludes our proof.
  
  \textbf{Case $\bm{n > 0}$.} We consider two subcases. (1) There exists
  a write $w \in \writeSet$ whose value does not appear in
  $\reads$. We consider the trace $\tr' = (\reads,\writeSet \setminus
  \set{w})$. By definition of $\MPC$, $\tr'$ belongs to $\MPC$, and we
  deduce by induction hypothesis that there exists an execution $e' \in
  \semantics{\impl}$ such that $\gettrace{e'} = \tr'$. By
  Lemma~\ref{lemma:updavailability}, we extend $e'$ in an execution $e
  \in \semantics{\impl}$ so that $\gettrace{e} = \tr$, which is what we
  wanted to prove.

  (2) All the writes of $\writeSet$ appear in the reads of $\reads$. By
  the Consistency and Prefix properties of $\MPC$, there exists a
  non-empty sequence $\ell \in \Nat^+$ of elements from~$\writeSet$,
  such that all read actions return a prefix of~$\ell$, and there 
  exist read actions that return the whole list~$\ell$.
  
  Let $\ell = \ell' \cdot d$, where $d \in \Nat$ is the last element
  of $\ell$. Let $\tr'$ be the trace $(\reads', \writeSet \setminus
  \set{d})$, such that $\reads'$ is the trace $\reads$ where every query
  action labelled by $\ell$ is replaced by a query action labelled
  by~$\ell'$, and implicitly, every query action labelled by any prefix
  of $\ell'$ is unchanged. Let $R$ the set of the newly added query
  actions, and let $P \subseteq \Pid$ be the set of site identifiers
  that appear in an action of~$R$.

  By definition of $\MPC$, we have $\tr' \in \MPC$. By induction 
  hypothesis, we deduce that there exists a finite execution $\exec' \in
  \semantics{\impl}$ such that $\gettrace{\exec'} = \tr'$. 
  
  Then, by Lemma~\ref{lemma:updavailability}, we add at the end
  of~$\exec'$ an update action (on some site $\pid \in \Pid$ and with
  some fresh $\aid \in \Nat$), which is of the form
  $(\aid,\pid,\append(d))$, so we get an execution $\exec'' \in
  \semantics{\impl}$ such that $\gettrace{\exec''} = (\reads',
  \writeSet \setminus \set{d} \cup \set{d}) = (\reads', \writeSet)$.

  Using Lemma~\ref{lemma:convergence}, we extend $\exec''$ in an
  execution $e'''$ by adding queries to the sites in~$P$, as many as
  were replaced by queries in $R$. Since $\impl \strictstronger \MPC$,
  and since by Lemma~\ref{lemma:convergence}, the answers to these
  queries must contain all the elements of $\ell$, we conclude that the
  only possible answer for all these queries is the entire list~$\ell$.
  
  Finally, we use Lemma~\ref{lemma:invisiblereads} to remove the queries
  $R$ from $e'''$, and we obtain an execution in
  $\semantics{\impl}$ whose trace is~$\tr$. \foorp
\end{proof}

We then extend Lemma~\ref{lem:finitetrace} to infinite executions.

\begin{theorem}
  \label{theorem:main}
Let $\impl$ be a convergent implementation.
Then,  $\impl$ is not strictly stronger than \MPC:
\(
  \impl \not \strictstronger \MPC
\).
\end{theorem}

\begin{proof}
Assume that $\impl$ is strictly stronger than $\MPC$ i.e.~$\impl
\strictstronger \MPC$. Our goal is to prove that $\MPC \stronger
\impl$ therefore leading to a contradiction. In terms of traces, we
want to prove that $\MPC \subseteq \gettrace{\semantics{\impl}}$.

Let $\tr = (\reads,\writeSet) \in \MPC$. We need to show that $\tr \in
\gettrace{\semantics{\impl}}$. 

 
\textbf{Case where $\bm\tr$ is finite.}
Proven in Lemma~\ref{lem:finitetrace}.

\textbf{Case where $\bm\tr$ is infinite.} Let $\reads =
(A,<,\labell)$. We first order all the query actions in $A$ as a
sequence $\aid_1,\dots,\aid_n,\dots$ such that for every $i \geq 1$,
$\labell(\aid_i) \leqp \labell(\aid_{i+1})$, and for every $i,j \geq
1$, $\aid_i < \aid_j$ (in the program order of $\reads$) implies $i <
j$.  Defining such a sequence is possible thanks to the Monotonicity
property of \MPC.

For each $i \geq 1$, we define a \emph{finite} trace $\tr_i$ that
contains all query actions $\aid_j$ with $j \leq i$, and the subset
$W_i$ of $\writeSet$ that contains all elements appearing in these query
actions, i.e.~$W_i = \set{d \in W \ |\ d \in \labell(\aid_i)}$. Our goal
is to construct an execution $e_i \in \semantics{\impl}$ such that
$\gettrace{e_i} = \tr_i$, and such that for all $i \geq 1$, $e_i \ltp
e_{i+1}$. We then define $e^\infty$ as the limit of
$e_1,\dots,e_n,\dots$ By Lemma~\ref{lemma:limit}, we have $e^\infty \in
\semantics{\impl}$. Since $\gettrace{e^\infty} = \tr$, we deduce that
$\tr \in \gettrace{\semantics{\impl}}$, which concludes the proof.

We now explain how to construct $e_i$, for every $i \geq 1$, by
induction on $i$. Let $e_0$ be the empty execution and $\tr_0 =
\gettrace{e_0}$. For $i \geq 0$, we define $e_{i+1}$ by starting from
$e_{i}$, and extending it as follows. By induction, we know that
$\gettrace{e_i} = \tr_i$, and want to extend it into an execution
$e_{i+1}$ such that $\gettrace{e_{i+1}} = \tr_{i+1}$.

The next step of the proof is similar to the proof of
Lemma~\ref{lemma:convergence}. For every broadcast action
$(\aid_1,\pid_1,\send{\msgid,\msg})$ in $e_i$, and for every $\pid_2
\in \Pid$, if there is no receive action
$(\_,\pid_2,\receive{\msgid,\msg})$ in $e_i$, then we add one when
constructing $e_{i+1}$. Moreover, if the message handler specifies
that a message $\msg\,'$ should be sent when $\msg$ is received, we
add a new broadcast action that sends $\msg\,'$, immediately following
the receive action.

Then, similarly to the construction in Lemma~\ref{lem:finitetrace}, we
add update and query actions (using Lemmas~\ref{lemma:updavailability},
\ref{lemma:availability}, and \ref{lemma:convergence}) in order to
obtain an execution $e_{i+1}$ such that $\gettrace{e_{i+1}} = \tr_{i+1}$.
\foorp
\end{proof}


\section{Comparison with Other Consistency Criteria}
\label{sec:related}

\paragraph{Relation between \MPC{} and other consistency criteria.}

Consistency criteria are usually defined in terms of \emph{full traces}
that contain both the read and write operations in the program order
(see e.g.,~\cite{principles-of-eventual-consistency}). The definition of
trace we used in this paper (Def.~\ref{sec:deftrace},
Section~\ref{sec:problem}) puts the writes in an unordered set,
unrelated to the read operations. This choice is justified in
large-scale, open, implementations, such as blockchain protocols.
Indeed, in these systems, any participant can perform a write operation
(e.g., a blockchain transaction), and the origin of the write has no
relevance for the protocol.


When considering full traces, \MPC as a consistency criterion is
strictly weaker than strong consistency. Indeed, \MPC allows a trace
where a read preceded by a write on the same site ignores that write.

As explained in the introduction, \MPC is not comparable to
causal consistency. \MPC allows full traces that causal consistency
forbids and vice versa. Therefore, our result stating that 
nothing stronger than \MPC that can be implemented in a distributed
setting does not contradict earlier results of~\cite{mahajan11cacTR}
and~\cite{attiya2017limitations}, which show that nothing stronger than
variants of causal consistency can be implemented.


\paragraph{Relation with other criteria when using our notion of a trace.}

When using our notion of a trace, \MPC is strictly stronger than causal
consistency. First, \MPC is stronger than causal consistency because
every trace of \MPC can be produced by a causally consistent system. The
main reason is that our notion of a trace does not capture any causality
relation. Moreover, there are some traces that causal consistency
produces and that do not belong to \MPC, e.g.~a trace where Site~1 has
a $\readlist[1,2]$ operation, Site~2 has a $\readlist[2,1]$, and where
$\append(1)$ and $\append(2)$ are not causally related as they happen 
at the \emph{same time} (this explains
that \MPC is \emph{strictly} stronger than causal consistency).

Moreover, it is interesting to note that, for our notion of a trace, the
traces allowed by \MPC are exactly the traces allowed by strong
consistency. This entails that, if the replicated data type is
used by clients that only have the observability defined by our traces, then there is no need
to implement strong consistency. In short, \MPC and strong consistency
are indistinguishable to these clients.

\section{Conclusion}
\label{sec:conclusion}

We have investigated the question of what is the strongest consistency
criterion that can be implemented when replicating a data structure,
in distributed systems under availability and partition-tolerance
requirements. Earlier work had established the impossibility of
implementing strong consistency in such a system model, but left open
the question of the strongest criteria that {\em can} be
implemented. In this paper we have focused on the \emph{Monotonic
  Prefix Consistency} (\MPC) criterion. We proposed an implementation
of \MPC and showed that no criterion stronger than \MPC can be
implemented.


It is worth noting that blockchain protocols, such as the Bitcoin
protocol~\cite{bbbitcoin}, implement \MPC with high probability: the
traces that the protocol produces are traces that belong to \MPC with
high probability. This was shown in
\cite{blockchainasync,btcbackbone}. More precisely, the authors proved that
the blockchains of two honest participants are compatible, in the
sense that one should be a prefix of the other with high probability,
when ignoring the last blocks\footnote{In Bitcoin-like protocols, the
  most recent blocks are ignored as they are considered unsafe to use
  until newer blocks are appended after them.}. This property is
called \emph{consistency} in \cite{blockchainasync}, and it
corresponds to the Prefix property we give in Section~\ref{sec:mpc}.
Moreover, it was shown~\cite{blockchainasync,btcbackbone} that the
blockchain of an honest participant only grows over time. This
property is called \emph{future-self consistency} in
\cite{blockchainasync}, and it corresponds to the Monotonicity
property we give in Section~\ref{sec:mpc}.

In future work we plan to investigate how the strongest achievable
consistency criterion depends on observability -- that is, the
information encoded in a trace -- and study conditions for the
(non)existence of a strongest consistency criterion.
We are also interested in extending our result to other system models.
Specifically, answering the question of what is the strongest
consistency criterion that can be implemented in systems where the
origin of updates do matter for the protocol. Also, the question whether
MPC is the strongest implementable consistency criterium in a {\em
probabilistic} setting, remains open.



\bibliographystyle{splncs}
\bibliography{references}

\begin{thebibliography}{10}

\bibitem{journals/toplas/HerlihyW90}
Herlihy, M., Wing, J.M.:
\newblock Linearizability: A correctness condition for concurrent objects.
\newblock ACM Trans. Program. Lang. Syst. \textbf{12}(3) (1990)

\bibitem{bre12cap}
Brewer, E.:
\newblock {CAP} twelve years later: How the ``rules'' have changed.
\newblock Computer \textbf{45}(2) (2012)

\bibitem{DBLP:journals/sigact/GilbertL02}
Gilbert, S., Lynch, N.A.:
\newblock Brewer's conjecture and the feasibility of consistent, available,
  partition-tolerant web services.
\newblock SIGACT News \textbf{33}(2) (2002)  51--59

\bibitem{attiya2017limitations}
Attiya, H., Ellen, F., Morrison, A.:
\newblock Limitations of highly-available eventually-consistent data stores.
\newblock IEEE Transactions on Parallel and Distributed Systems \textbf{28}(1)
  (2017)  141--155

\bibitem{Lamport:1978:TCO:359545.359563}
Lamport, L.:
\newblock Time, clocks, and the ordering of events in a distributed system.
\newblock Commun. ACM \textbf{21}(7) (July 1978)  558--565

\bibitem{replicated-data-consistency-explained-through-baseball}
Terry, D.:
\newblock Replicated data consistency explained through baseball.
\newblock Technical Report MSR-TR-2011-137, Microsoft Research (October 2011)

\bibitem{guerraoui2016trade}
Guerraoui, R., Pavlovic, M., Seredinschi, D.A.:
\newblock Trade-offs in replicated systems.
\newblock IEEE Data Engineering Bulletin \textbf{39} (2016)  14--26

\bibitem{principles-of-eventual-consistency}
Burckhardt, S.:
\newblock Principles of Eventual Consistency.
\newblock Now Publishers (October 2014)

\bibitem{DBLP:conf/popl/BouajjaniEH14}
Bouajjani, A., Enea, C., Hamza, J.:
\newblock Verifying eventual consistency of optimistic replication systems.
\newblock In Jagannathan, S., Sewell, P., eds.: The 41st Annual {ACM}
  {SIGPLAN-SIGACT} Symposium on Principles of Programming Languages, {POPL}'14,
  San Diego, CA, USA, January 20-21, 2014, {ACM} (2014)  285--296

\bibitem{mahajan11cacTR}
Mahajan, P., Alvisi, L., Dahlin, M.:
\newblock Consistency, availability, convergence.
\newblock Technical Report TR-11-22, Computer Science Department, University of
  Texas at Austin (May 2011)

\bibitem{bbbitcoin}
Nakamoto, S.:
\newblock Bitcoin: A peer-to-peer electronic cash system (2008)

\bibitem{blockchainasync}
Pass, R., Seeman, L., Shelat, A.:
\newblock Analysis of the blockchain protocol in asynchronous networks.
\newblock In: Annual International Conference on the Theory and Applications of
  Cryptographic Techniques, EUROCRYPT'17. Volume 10211 of Lecture Notes in
  Computer Science., Paris, France (April 2017)  643--673

\bibitem{btcbackbone}
Garay, J., Kiayias, A., Leonardos, N.:
\newblock The bitcoin backbone protocol: Analysis and applications.
\newblock In: Annual International Conference on the Theory and Applications of
  Cryptographic Techniques, Springer (2015)  281--310

\end{thebibliography}

\newpage
\appendix
\section{Proof of Feasability of MPC}
\label{ann:MPC}

\feasability*

\begin{proof}
  Let $e \in \semantics{\impl}$, we establish an inductive invariant
  that holds for every finite prefix $e'$ of $e$. Let $A$ be the set of
  action identifiers of $e'$. Let $\ell \in \Nat^*$ be the sequence of
  values that appear in a broadcast message $\texttt{Apply}$ from
  Site~$1$, in the order they appear in $e'$.

  Let $\Pid$ be the set of process identifiers. For each site $\pid \in
  \Pid$, consider the unique run $r$ for the projection
  $\project{e'}{\pid}$, and let $\ell_\pid \in \Nat^*$ be the sequence
  maintained in the local state of Site~$\pid$ at the end of the run $r$. 

  Let $\tr' = \gettrace{e'}$ be the trace of $e'$, with $\tr' =
  (\reads,\writeSet)$.

  Then, we have the following properties.
  \begin{enumerate}
  \item For every $\texttt{Apply}(d)$ message with $d \in \Nat$ that
      appears in $e'$ (from Site~$1$), we have $d \in W$.

  \item For every $\texttt{Forwarded}(d)$ message with $d \in \Nat$ that
      appears in $e'$ (from Site~$i$ with $i > 1$), we have $d \in W$.

  \item The elements of $\ell$ are in $W$.

  \item For every $\pid \in \Pid$, $l_\pid \leqp \ell$.

  \item For every query $(\_,\pid,\readlist\ \ell')$ in $e$ with $\pid \in
      \Pid$, we have $\ell' \leqp \ell_\pid$.

  \item \textbf{Consistency:} For all $\aid \in A$, and for any element
      $d \in \Nat$ of $\labell(\aid)$, we have $d \in \writeSet$.

  \item \textbf{Prefix:} For any all $\aid,\aid\,' \in A$, $\labell(\aid)
      \leqp
      \labell(\aid\,')$ or $\labell(\aid\,') \leqp \labell(\aid)$.

  \item \textbf{Monotonicity:} For all $\aid,\aid\,' \in A$, if $\aid <
      \aid\,'$, then $\labell(\aid) \leqp
      \labell(\aid\,')$.
  \end{enumerate}
  
  We can see that this invariant holds for the empty execution, and that
  any action that the implementation can take maintains it. 
\end{proof}

\section{Closure Properties of Implementations}
\label{ann:closure}

\queryavailability*

\begin{proof}
  Similar to the proof of Lemma~\ref{lemma:updavailability}, but
  using the query handler, instead of the update handler. This proof is
  also simpler, as there is no need to consider messages, since the
  query handler cannot broadcast any message. Therefore, in this proof,
  only case~1 needs to be considered. \foorp
\end{proof}

\invisiblereads*

\begin{proof}
  This is a direct consequence of Definition~\ref{def:follows}, which
  specifies that query actions do not modify the local state of sites,
  and do not broadcast messages. \foorp
\end{proof}

\limitexec*

\begin{proof}
  According to Definition~\ref{def:set-of-exec}, we have three points
  to prove.
  
  (1) (Projection) First, we want to show that, for all $\pid \in \Pid$,
  the projection $\project{e^\infty}{\pid}$ follows~$\impl$. For all $i
  \geq 1$, we know that $e_i \in \semantics{\impl}$, and deduce that
  $\project{e_i}{\pid}$ follows~$\impl$. Let $r_i$ be the run of
  $\project{e_i}{\pid}$. Note that for all $i \geq 1$, we have $r_i \ltp
  r_{i+1}$. Let $r^\infty_\pid$ be the limit of the runs
  $r_1,\dots,r_n,\dots$ By construction, $r^\infty_\pid$ is a run of
  $\project{e^\infty}{\pid}$, which shows that
  $\project{e^\infty}{\pid}$ follows~$\impl$.

  (2) (Causality) We need to prove that every receive action $\act$ in
  $e^\infty$ has a corresponding broadcast action $\act'$ that precedes
  it in~$e^\infty$. Let $e_i$ be a prefix of $e^\infty$ that
  contains~$\act$. Since $e_i \in \semantics{\impl}$, we know that there
  exists a broadcast action $\act'$ corresponding to~$\act$, and that
  precedes $\act$ in~$e_i$. Finally, since $e_i \ltp e^\infty$, $\act'$
  precedes $\act$ in~$e^\infty$.

  (3) (Fairness) We want to prove that for every broadcast action
  $\act$ of $e^\infty$ and for every site $\pid \in \Pid$, there
  exists a corresponding receive action~$\act'$. Let $e_i$ be a prefix
  of $e^\infty$ that contains~$\act$. By assumption of the current
  lemma, there exists $j \geq 1$ such that $e_j$ contains a receive
  action $\act'$ corresponding to~$\act$. Moreover, since $e_j \ltp
  e^\infty$, $\act'$ belongs to~$e^\infty$, which concludes our proof.
  \foorp
\end{proof}

\end{document}